\documentclass[11pt]{article}
\usepackage{geometry}
\usepackage{amsmath, amssymb, amsthm}
\usepackage{fullpage}
\usepackage{hyperref}
\hypersetup{colorlinks = true, linkcolor=blue, citecolor=blue}
\usepackage{cleveref}
\usepackage{algorithm}
\usepackage[noend]{algpseudocode}
\usepackage{xspace}
\usepackage{tikz}
\usetikzlibrary{shapes.geometric}

\newtheorem{theorem}{Theorem}
\newtheorem{claim}[theorem]{Claim}
\newtheorem{definition}[theorem]{Definition}
\newtheorem{lemma}[theorem]{Lemma}
\newtheorem*{claim*}{Claim}
\newtheorem{observation}[theorem]{Observation}

\title{Fault-tolerant $k$-Supplier with Outliers}
\author{
Deeparnab Chakrabarty\thanks{Dartmouth College. Emails: \texttt{\{deeparnab, luc.cote.th, ankita.sarkar.gr\}@dartmouth.edu}. DC and AS were partly supported by NSF Award \#2041920.}
\and
Luc Cote\footnotemark[1]
\and
Ankita Sarkar\footnotemark[1]
}
\date{}

\newcommand{\fkso}{{\sf F}$k${\sf SO}\xspace}
\newcommand{\fksofull}{Fault-tolerant $k$-Supplier with Outliers\xspace}
\newcommand{\ufkso}{{\sf UF}$k${\sf SO}\xspace}
\newcommand{\ufksofull}{Uniformly Fault-tolerant $k$-Supplier with Outliers\xspace}
\newcommand{\kso}{$k${\sf SO}\xspace}
\newcommand{\ksofull}{$k$-Supplier with Outliers\xspace}
\newcommand{\ks}{$k$-Supplier\xspace}
\newcommand{\fks}{{\sf F}$k${\sf S}\xspace}
\newcommand{\fksfull}{Fault-tolerant $k$-Supplier\xspace}

\Crefrangeformat{claim}{Claims #3#1#4-#5#2#6}

\DeclareMathOperator{\argmax}{\mathrm{argmax}}
\newcommand{\child}{\mathsf{child}}
\newcommand{\cov}{\mathsf{cov}}
\newcommand{\height}{\mathrm{height}}

\newcommand{\opt}{\mathtt{opt}}
\newcommand{\poly}{\mathrm{poly}}
\newcommand{\Roots}{\mathtt{Roots}}

\newcommand{\calC}{\mathcal{C}}
\newcommand{\calF}{\mathcal{F}}

\newcommand{\calP}{\mathcal{P}}

\newcommand{\calT}{\mathcal{T}}

\newcommand{\NN}{\mathbb{N}}
\newcommand\RR{\mathbb{R}}

\newcommand\bone{\mathbf{1}}

\newcommand\bk{\mathbf{k}}

\newcommand{\abs}[1]{\left\lvert#1\right\rvert}
\newcommand{\bracket}[1]{\left[#1\right]}
\newcommand{\ceil}[1]{\left\lceil#1\right\rceil}
\newcommand{\floor}[1]{\left\lfloor#1\right\rfloor}
\newcommand\paren[1]{\left(#1\right)}
\newcommand{\set}[1]{\left\{#1\right\}}

\begin{document}

\maketitle

\begin{abstract}
\pagenumbering{gobble}
We present approximation algorithms for the Fault-tolerant $k$-Supplier with Outliers ({\sf F$k$SO}) problem. This is a common generalization of two known problems -- $k$-Supplier with Outliers, and Fault-tolerant $k$-Supplier -- each of which generalize the well-known $k$-Supplier problem. In the $k$-Supplier problem the goal is to serve $n$ clients $C$, by opening $k$ facilities from a set of possible facilities $F$; the objective function is the farthest that any client must travel to access an open facility. In {\sf F$k$SO}, each client $v$ has a \emph{fault-tolerance} $\ell_v$, and now desires $\ell_v$ facilities to serve it; so each client $v$'s contribution to the objective function is now its distance to the $\ell_v$\textsuperscript{th} closest open facility. Furthermore, we are allowed to choose $m$ clients that we will serve, and only those clients contribute to the objective function, while the remaining $n-m$ are considered outliers.

Our main result is a $(4t-1)$-approximation for the {\sf F$k$SO} problem, where $t$ is the number of distinct values of $\ell_v$ that appear in the instance. At $t=1$, i.e. in the case where the $\ell_v$'s are uniformly some $\ell$, this yields a $3$-approximation, improving upon the $11$-approximation given for the uniform case by Inamdar and Varadarajan [2020], who also introduced the problem. Our result for the uniform case matches tight $3$-approximations that exist for $k$-Supplier, $k$-Supplier with Outliers, and Fault-tolerant $k$-Supplier.

Our key technical contribution is an application of the \emph{round-or-cut} schema to {\sf F$k$SO}. Guided by an LP relaxation, we reduce to a simpler optimization problem, which we can solve to obtain distance bounds for the ``round'' step, and valid inequalities for the ``cut'' step. By varying how we reduce to the simpler problem, we get varying distance bounds -- we include a variant that gives a $(2^t + 1)$-approximation, which is better for $t \in \set{2,3}$. In addition, for $t=1$, we give a more straightforward application of round-or-cut, yielding a $3$-approximation that is much simpler than our general algorithm.

\end{abstract}
\newpage

\section{Introduction}\label{sec:intro}
\pagenumbering{arabic}
Clustering problems form a class of discrete optimization problems that appear in many application areas ranging from operations research~\cite{Lee1981,Mirkin1996,Liu1999} to machine learning~\cite{JainMF1999,AhujaCGAK2020,MahabV2020,JungKL2020}. They also 
have formed a sandbox where numerous algorithmic ideas, especially ideas in approximation algorithms, have arisen and developed over the years. 
One of the first clustering problems to have been studied is the \ks problem~\cite{HochbS1986}: in this problem, one is given a set of points in a metric space $(C\cup F,d)$, where $C$ is the set of ``clients'' and $F$ is the set of ``facilities'', and a number $k$. The objective is to ``open'' a collection $S\subseteq F$ of $k$ centers so as to minimize the maximum distance between a client $v\in V$ to its nearest open center in $S$, that is, minimize $\max_{v\in C} \min_{f\in S} d(f,v)$. It has been known since the mid-80's, due to an influential paper of Hochbaum and Shmoys~\cite{HochbS1986}, that this problem has a $3$-approximation and no better approximation is possible\footnote{Howard Karloff is attributed the hardness result in~\cite{HochbS1986}.}.

One motivation behind the objective function above is that $d(v,S) := \min_{f\in S} d(f,v)$ indicates how (un)desirable the client $v$ perceives the the set of open facilities, and the \ks objective tries to take the egalitarian view of trying to minimize the unhappiest client. However, in certain applications, a client $v$ would perhaps be interested not only in having one open facility in a small neighborhood but a larger number. For instance, the client may be worried about some open facilities closing down. This leads to the {\em fault-tolerant} versions of clustering problems. In this setting, each client $v$ has an integer $\ell_v$ associated with it, and the desirability of a subset $S$ for $v$ is not determined by the nearest facility in $S$, but rather the $\ell_v$\textsuperscript{th} nearest facility. That is, we sort the facilities in $S$ in increasing order of $d(f,v)$ and let $d_{\ell_v}(v,S)$ denote the $\ell_v$\textsuperscript{th} distance in this order
(so $d(v,S) = d_1(v,S)$). The \fksfull (\fks) problem is to find $S\subseteq F$ with $|S| = k$ so as to minimize $\max_{v\in C} d_{\ell_v} (v,S)$.
As far as we know, the \fksfull problem has not been {\em explicitly} studied in the literature\footnote{although the fault-tolerant facility location and $k$-median have been extensively studied~\cite{JainV2004, GuhaMM2003, SwamyS2008, HajiaHLLS2016}; more on this in~\Cref{sec:related-work}.}, however, as we show in~\Cref{sec:prelim}, there is a simple $3$-approximation based on the same scheme developed by Hochbaum and Shmoys~\cite{HochbS1986}. 

One drawback of the \ks objective is that it is extremely sensitive to outliers; since one is trying to minimize the maximum, a single far-away client makes the optimal value large. To allay this, people have considered the ``outlier version'' of the problem, \ksofull (\kso). In \kso, one is given an additional integer parameter $m$, and the goal of the algorithm is to open a subset $S$ of $k$ facilities {\em and} recognize a subset $T\subseteq C$ of $m$-clients, so as to minimize $\max_{v\in T} d(v,S)$. That is, all clients outside $T$ are deemed outliers and one doesn't consider their distance to the solution. 
The outlier version is algorithmically interesting and is not immediately captured by the Hochbaum-Shmoys technique. Nevertheless, in 2001, Charikar, Khuller, Mount, and Narsimhan~\cite{ChariKMN2001} described a combinatorial, greedy-like $3$-approximation for \kso\footnote{A different LP-based approach was taken by Chakrabarty, Goyal, and Krishnaswamy~\cite{ChakrGK2020} and vastly generalized by Chakrabarty and Negahbani~\cite{ChakrN2019}; more on this in~\Cref{sec:related-work,sec:prelim}.}. Since then, outlier versions of many clustering problems have been considered, and it has been a curious feature that the approximability of the outlier version has been of the same order as the approximability of the original version without outliers.

In this paper, as suggested by the title, we study the {\em \fksofull} (\fkso) problem which generalizes \fks and \kso. This problem was explicitly studied only recently by Inamdar and Varadarajan~\cite{InamdV2020}; but that work only studies the {\em uniformly fault-tolerant} case where all $\ell_v$'s are the same (say, $\ell$). The main result of~\cite{InamdV2020} was a ``reduction'' to the ``non-fault-tolerant''  version of the clustering problem with outliers, and their result is that an $\alpha$-approximation for the \kso problem translates\footnote{Actually, they~\cite{InamdV2020} only study the ``$k$-Center'' case when $F=C$, and in that case the result is $(2\alpha +2)$; their proofs do reveal that for the Supplier version, one obtains $(3\alpha + 2)$.} to a $(3\alpha + 2)$-approximation for the \fkso problem with uniform fault-tolerance. Setting $\alpha = 3$ from the aforementioned work~\cite{ChariKMN2001} on \kso, one gets an $11$-approximation for the uniform case of \fkso.

\paragraph*{Our Contributions}

We begin by providing a simple LP-based $3$-approximation for the \fkso problem when the fault-tolerances are uniform, that is, $\ell_v = \ell$ for all $v\in C$. This improves the known $11$-approximation ~\cite{InamdV2020}. Even this special case is interesting in that when the uniform fault-tolerance $\ell$ divides $k$, then the ``natural LP'' suffices to obtain a $3$-approximation using a rounding scheme similar to a prior rounding algorithm for \kso~\cite{ChakrGK2020}. However, when $\ell$ doesn't divide $k$, then we need to add in valid inequalities akin to Chv\'{a}tal-Gomory cuts~\cite{Chvat1973,Gomor1969} in integer programming. Nevertheless, the rounding algorithm is simple and is described in~\Cref{sec:ufkso}.

Our main contribution is to the general \fkso problem, when $\ell_v$'s can be different for different clients. This problem becomes much more complex for the simple reason that if two clients $v$ and $v'$ are located very close together, but $\ell_v < \ell_{v'}$, then opening $\ell_v$ facilities around $v$ would still render $v'$ unhappy -- this does not happen in the uniform case. Therefore, the Hochbaum-Shmoys procedure~\cite{HochbS1986}, or more precisely the LP-guided Hochbaum-Shmoys rounding that is known for \kso~\cite{ChakrGK2020}, simply doesn't apply under non-uniform fault-tolerance. Indeed, the natural LP relaxation and its natural strengthening, which give us the $3$-approximation for the uniform case, has large integrality gaps even when the $\ell_v$'s take only two values; we show this in~\Cref{sec:weak-lp-gap}.

Our main result is a $(4t-1)$-approximation for the \fkso problem when there are $t$ distinct\footnote{We should point to the reader that one can't simply solve $t$ different uniform \fkso versions and ``stick them together'' to get such a result; although it is a natural idea, note that a priori we do not know how many outliers
one will obtain from each fault-tolerant class, and enumerating is infeasible.} $\ell_v$'s (that is, $\abs{\{\ell_v: v\in C\}} = t$). When $t=1$, we recover the $3$-approximation mentioned above. This is not the most desirable result (one would hope a $O(1)$-approximation for any $t$), but as the above integrality gap example illustrates, even when $t=2$, strong LPs have bad integrality gaps. We also use the same schema to give a $(2^t + 1)$-approximation, which gives better approximation factors for $t\in\set{2, 3}$. 

Our main technical contribution is to apply the {\em round-or-cut} schema introduced for clustering problems by Chakrabarty and Negahbani~\cite{ChakrN2019} to \fkso. In particular, this schema uses a fractional solution $\set{\cov_v}_{v \in C}$ which indicates the extent to which each client $v$ is an ``inlier'' (that is, in the final set $T$ of at least $m$ clients). Earlier works~\cite{ChakrGK2020, ChakrN2019} use this fractional solution to guide the Hochbaum-Shmoys-style~\cite{HochbS1986} rounding algorithm, creating a partition on the set of clients and solving a simpler optimization problem on this partition. We also use the same schemata, except that our partitioning scheme is a more general one warranted by the non-uniform fault-tolerances; nevertheless, we show that we either obtain the desired approximation factor (the ``round'' step), or we can prove that the $\cov_v$'s cannot arise as a combination of integral solutions (the ``cut'' step). Once we do this, the round-or-cut schema implies a polynomial time approximation factor. 
We also show that $t$ is the limiting factor in our approach; more precisely, the diameter of the parts of the desired partition dictates the upper bound on the approximation factor, and in \Cref{appsec:partition-gap} we construct an instance such that the diameter needs to be $\Omega(t)$. We leave the possibility of obtaining $O(1)$-approximations for \fkso, or alternatly proving a super-constant hardness, as an intriguing open problem.

\subsection{Related Work}\label{sec:related-work}
The Hochbaum-Shmoys algorithm~\cite{HochbS1986} gives a $3$-approximation for the \ks problem, and has been extended to give approximation algorithms for multiple related problems. Plesn\'{i}k~\cite{Plesn1987} gave one such extension, obtaining a $3$-approximation\footnote{The work~\cite{Plesn1987} studies the $k$-Center case, where $F = C$, and gives a $2$-approximation; but the proofs imply a $3$-approximation for the Supplier version.} when each client $v$ has weight $w(v)$, and this scales the client's ``unhappiness'', so that the objective function becomes $\max_{v \in C}\paren{w(v)\cdot\min_{f \in S}d(v,f)}$. In another direction, Chakrabarty, Goyal, and Krishnaswamy~\cite{ChakrGK2020} gave an extension to \ksofull, using an LP relaxation to indicate which clients are outliers, and obtaining a Hochbaum-Shmoys-like $3$-approximation. This was vastly extended by Chakrabarty and Negahbani~\cite{ChakrN2019}, implying a $3$-approximation for multiple problems, including \kso with knapsack constraints on the facilities. Bajpai, Chekuri, Chakrabarty, and Negahbani~\cite{BajpaCCN2021} generalized the aforementioned weighted version~\cite{Plesn1987} to handle outliers, matroid constraints, and knapsack constraints, obtaining constant approximation ratios for each.

In the early 2000s, Jain and Vazirani~\cite{JainV2004} introduced the notion of fault-tolerance for the Uncapacitated Facility Location (UFL) problem. The notion has thereafter been studied for various related problems: UFL~\cite{GuhaMM2003,SwamyS2008,ByrkaSS2010}; UFL with multiset solutions, often called facility \emph{placement} or \emph{allocation}~\cite{XuS2009,YanC2011,RybicB2015}; $k$-Median~\cite{SwamyS2008,KumarR2013,HajiaHLLS2016}; matroid and knapsack Median~\cite{Deng2022}; and $k$-Center~\cite{Krumke95,ChauGR98,KhullPS2000,KumarR2013}. In particular relevance to this paper, the \fkso problem was studied by Inamdar and Varadarajan~\cite{InamdV2020}. In addition, prior work also addresses alternate notions of fault-tolerance and outlier-type constraints. In a 2020 preprint, Deng~\cite{Deng2020} combines fault-tolerance with an outlier-type constraint requiring that the number of \emph{client-facility connections}, rather than the weight of satisfied clients, be at least some $m$. An altogether different notion of fault-tolerance has also been studied~\cite{ChechP2014,LiXDX2012, SonarSX2023}, where clients each want just one facility, but an adversary secretly causes some $k' \leq k$ of the chosen facilities to fail.

The round-or-cut schema that this paper applies, has found widespread usage in clustering problems, including in problems related to \ks. For example, the weighted version of \ks~\cite{Plesn1987,BajpaCCN2021} can be extended to impose different budgets to different weight classes -- i.e. there is no longer one $k$, but one $k_i$ per distinct weight $w_i$. This version admits a constant-factor approximation for certain special cases~\cite{ChakrN2023,JiaRSS2022} via the round-or-cut schema. Round-or-cut has also been used for \ks with covering constraints~\cite{AneggAKZ2022}, and for the Capacitated Facility Location problem~\cite{AnSS2014}. In the continuous clustering realm, where facilities can be picked from a potentially infinite-sized ambient metric space, round-or-cut has been used to circumvent the infinitude of the instance~\cite{ChakrNS2022}.

\section{Preliminaries}\label{sec:prelim}
Before we formally define our main problem, let us set up some important notation. 
\begin{definition}\label{def:dlv-nlv}
    Given a subset $S\subseteq F$, a client $v\in C$, and $a\in [k]$, 
    let $d_{a}(v,S)$ be the distance of $v$ to its $a$\textsuperscript{th} closest neighbor in $S$ (breaking ties arbitrarily and consistently). So $d_1(v,S) = d(v,S)$. Also let
    $N_a(v,S) \subseteq S$ denote the $a$ facilities in $S$ that are closest to $v$.
\end{definition}

\begin{definition}[\fksfull and \fksofull]
    In the \emph{\fksfull (\fks)} problem, we are given a finite metric space $(C \cup F,d)$, where $C$ is a set of $n$ clients and $F$ is a set of $\poly(n)$ facilities. We are also given a parameter $k \in \NN$, and \emph{fault-tolerances} $\set{\ell_v\in [k]}_{v \in C} $. The goal is to open $k$ facilities, i.e. pick $S \subseteq F : \abs S \leq k$, minimizing $\max_{v \in C}d_{\ell_v}(v,S)$.

    In the \emph{\fksofull (\fkso)} problem, we are given an \fks instance along with an additional parameter $m \in [n]$. The goal is to pick $S$ of size $k$ as before, along with \emph{inliers} $T \subseteq C : \abs T \geq m$, minimizing $\max_{v \in T}d_{\ell_v}(v,S)$.
\end{definition}
In the absence of fault-tolerances and outliers, i.e. in the \ks problem, the Hochbaum-Shmoys algorithm~\cite{HochbS1986} achieves a $3$-approximation as follows. It starts with a guess $r$ of the optimum value, large enough that every client $j$ has a facility within distance $r$ of itself, but otherwise arbitrary. Then it picks an arbitrary client $j$, opens a facility within distance $r$ of $j$, and deletes the set of ``children'' of $j$, which is $\child(j) := B(j,2r) \cap C = \set{v \in C: d(v,j) \leq 2r}$. Then it repeats this with the remaining clients, until there are no clients left. Observe that the $j$'s picked over the iterations -- call them the set $R$ -- has the following \emph{well-separated} property.
\begin{definition}[$r$-well-separated set]
    A set $X \subseteq C$ is \emph{$r$-well-separated} if for distinct $x,y \in X$, we have $d(x,y) > 2r$. Where $r$ is clear from context, we simply say that $X$ is \emph{well-separated}.
\end{definition}
Since $R$ is well-separated, it takes $\abs R$ clients to provide every $j \in R$ with a facility in $B(j,r)$. So if $\abs R > k$, then the guess of $r$ is too small -- we can double $r$ and retry the algorithm. On the other hand, if $\abs R \leq k$, then the guess is either correct or too large, so we halve $r$ and retry. This binary search yields the correct $r$, and the following guarantee: $\set{\child(j)}_{j \in R}$ partitions $C$, and for a $v \in \child(j)$, since $d(v,j) \leq 2r$ and we opened a facility in $B(j,r)$, there is a facility within distance $3r$ of $v$. This means that we have a $3$-approximation.

The Hochbaum-Shmoys algorithm described above, generalizes to give a $3$-approximation for \fks via the following modifications: instead of picking $j$'s into $R$ in arbitrary order, we pick them in decreasing order of $\ell_v$'s; we also open $\ell_j$ facilities in each $B(j,r) : j \in R$, instead of just one. This guarantees that, if $v \in \child(j)$, $\ell_v \leq \ell_j$, allowing us to extend the Hochbaum-Shmoys~\cite{HochbS1986} guarantee to \fks. We now formally state this algorithm.

\begin{algorithm}[H]\caption{Hochbaum-Shmoys~\cite{HochbS1986} modified for \fks}\label{alg:fks}
    \begin{algorithmic}[1]
        \Require $C$
        \State $U \gets C$
        \State $R \gets \emptyset$
        \State $S \gets \emptyset$
        \While{$U \neq \emptyset$}
            \State $j \gets \argmax_{v \in U}\ell_v$\label{alg:fks:ln:j}
            \State $R \gets R \cup \set j$
            \State $i_1,i_2,\dots,i_{\ell_j} \gets \ell_j$ arbitrary facilities in $B(j,r) \cap F$\label{alg:fks:open}\Comment{they exist by choice of $r$}
            \State $S \gets S \cup \set{i_1,i_2,\dots,i_{\ell_j}}$
            \State $\child(j) \gets \set{v \in U : d(v,j) \leq 2r}$
            \State $U \gets U \setminus \child(j)$
        \EndWhile
        \Ensure $S \subseteq F$
    \end{algorithmic}
\end{algorithm}
To show that this algorithm yields a $3$-approximation, we need to argue that $\forall v \in C$, $d_{\ell_v}(v,S) \leq 3r$. To see this, consider $j \in R: v \in \child(j)$. \Cref{alg:fks:ln:j} guarantees that $\ell_v \leq \ell_j$, so $d_{\ell_v}(v,S) \leq d_{\ell_j}(v,S)$. By triangle inequalities, this is at most $d(v,j) + d_{\ell_j}(j,S)$. By construction of $\child(j)$, $d(v,j) \leq 2r$; and since we open $\ell_j$ facilities in Line~\ref{alg:fks:open}, $d_{\ell_j}(j,S) \leq r$. We have just shown that
\begin{theorem}\label{thm:fks}
    The \fks problem admits a $3$-approximation.
\end{theorem}

One way to achieve a $3$-approximation for the \ksofull problem, described in~\cite{ChakrGK2020}, is as follows: under a guess of $r$ as before, a linear program relaxation is used to assign variables $\cov_v \in [0,1]$ to each client $v$, representing whether $v$ is ``covered'', i.e. whether there is an open facility in $B(v,r)$. The LP-guided Hochbaum-Shmoys algorithm considers clients in decreasing order of these $\cov_v$'s, and we wait to pick facilities until the loop terminates. Then, facilities are opened near those $j \in R$ that have the $k$ largest $\abs{\child(j)}$. The LP relaxation is used to ensure that $\geq m$ clients are served in this way. This {\em does not} generalize directly to \fkso because the decreasing order of $\ell_v$'s that we employed above for \fks can conflict with the decreasing order of $\cov_v$'s (indeed, one may just expect clients $v$ with large $\ell_v$'s would be more likely to be outliers, that is, have low $\cov_v$'s). So in our algorithm for \fkso, we elect to follow the $\cov_v$ order, and explicitly force $\ell_v \leq \ell_j$ for $v$'s that we pick into $\child(j)$. This choice breaks the well-separated property of $R$, so our techniques are devoted to obtaining other well-separated sets that can guide our rounding; details of this can be found in~\Cref{sec:fkso}. When all the $\ell_v$'s are the same, though, we can indeed use the natural LP relaxation (with a slight strengthening to take care of divisibility issues), and we show this in the next section.

\section{\texorpdfstring{$3$-approximation for \ufkso}{3-approximation for UFkSO}}\label{sec:ufkso}
In this section, we address the \emph{uniform} case, where all fault-tolerances in the instance are the same, i.e.
\begin{definition}[\ufksofull (\ufkso)]
    The \ufksofull problem is a special case of the \fkso problem where, for an $\ell \in \NN$, $\forall v \in C$, $\ell_v = \ell$. 
\end{definition}
We prove that
\begin{theorem}\label{thm:ufkso-3-appx}
    The \ufkso problem admits a $3$-approximation.
\end{theorem}

Our algorithm begins by rounding a solution to the following LP relaxation, closely mimicking the $3$-approximation for \kso~\cite{ChakrGK2020} described in the last paragraph of the previous section. This rounding suffices when $\ell \mid k$. When $\ell \nmid k$, we identify a valid inequality for the round-or-cut framework. In the LP, the variables $\set{\cov_v}_{v \in C}$ denote whether or not a client $v \in C$ is \emph{covered} i.e. served within distance $r$; and variables $\set{x_i}_{i \in F}$ denote whether or not a facility $i \in F$ is open. $B(v,r)$ is the \emph{ball} of radius $r$ around $v$, containing all points within distance $r$ of $v$, i.e. $B(v,r) := \set{x \in C \cup F : d(v,x) \leq r}$.
\begin{align}
    &\sum_{v \in C}\cov_v \geq m\label{lp:fkso-weak:m}\tag{WL1}\\
    &\sum_{i \in F}x_i \leq k\label{lp:fkso-weak:k}\tag{WL2}\\
    \forall v \in C,\quad &\sum_{i \in F \cap B(v,r)}x_i \geq \ell\cov_v\label{lp:fkso-weak:opt-ball}\tag{WL3}\\
    \forall v \in C : d_\ell(v,F) > r,\quad &\cov_v = 0\label{lp:fkso-weak:inliers}\tag{WL4}\\
    \forall v \in C,\,i\in F,\quad &0 \leq \cov_v, x_i \leq 1\label{lp:fkso-weak:bounds}\tag{WL5}
\end{align}
Here, \eqref{lp:fkso-weak:m} enforces that at least $m$ clients must be covered, and \eqref{lp:fkso-weak:k} enforces that at most $k$ facilities can be opened. \eqref{lp:fkso-weak:opt-ball} and \eqref{lp:fkso-weak:inliers} connect the $\cov_v$ variables with the $x_i$ variables, ensuring that a client cannot be covered unless there are sufficient facilities opened within distance $r$ of it. Finally, \eqref{lp:fkso-weak:bounds} enforces that a client can be covered only once, and a facility can be opened only once. \Cref{clm:fkso:weak-lp-validity} shows that this LP is a valid relaxation of our problem. We defer its proof to \Cref{appsec:omitted-proofs}.
\begin{claim}\label{clm:fkso:weak-lp-validity}
    An instance of \ufkso is feasible iff it admits an integral solution satisfying \eqref{lp:fkso-weak:m}-\eqref{lp:fkso-weak:bounds}.
\end{claim}

Given a solution $\paren{\set{\cov_v}_{v \in C}, \set{x_i}_{i \in F}}$ satisfying \eqref{lp:fkso-weak:m}-\eqref{lp:fkso-weak:bounds}, we round as per \Cref{alg:ufkso-3-appx}. This algorithm constructs a well-separated set of \emph{representatives} $R_\cov \subseteq C$. Each client $v$ that has $\cov_v > 0$ becomes the \emph{child} of some representative, yielding a partition $\set{\child(j)}_{j \in R_\cov}$ of these clients. Then, facilities $S_\cov \subseteq F$ are opened in a manner that serves the $\floor{\frac k \ell}$ largest $\child(j)$ sets within distance $3r$.

\begin{algorithm}[H]\caption{$3$-approximation for \ufkso\label{alg:ufkso-3-appx}}
    \begin{algorithmic}[1]
        \Require $\paren{\set{\cov_v}_{v \in C}, \set{x_i}_{i \in F}}$ satisfying \eqref{lp:fkso-weak:m}-\eqref{lp:fkso-weak:bounds}
        \State $R_\cov \gets \emptyset$\label{alg:ufkso-3-appx:ln:filtering-start}
        \State $U \gets \set{v \in C: \cov_v > 0}$\label{alg:ufkso-3-appx:ln:positive-cov}
        \While{$U \neq \emptyset$}\Comment{filtering}\label{alg:ufkso-3-appx:ln:reps-loop}
            \State $j \gets \argmax_{v \in U}\cov_v$\label{alg:ufkso-3-appx:ln:max-cov}
            \State $R_\cov \gets R_\cov \cup \set v$
            \State $\child(j) \gets B(j,2r)\cap U$\label{alg:ufkso-3-appx:ln:2opt-ball}
            \State $U \gets U \setminus \child(j)$\label{alg:ufkso-3-appx:ln:disjoint}
        \EndWhile\label{alg:ufkso-3-appx:ln:filtering-end}
        \State $S_\cov \gets \emptyset$ \label{alg:ufkso-3-appx:ln:opening-start}
        \State $R' \gets R_\cov$
        \While{$\abs{S_\cov} < \floor{\frac k \ell}\cdot\ell$}
    \label{alg:ufkso-3-appx:ln:s-loop}\Comment{picking facilities to open}
            \State $j \gets \argmax_{j' \in R'}\abs{\child(j')}$\label{alg:ufkso-3-appx:ln:u-choose}
            \State $R' \gets R' \setminus \set j$
            \State $S_\cov \gets S_\cov \cup N_\ell(j,F)$\label{alg:ufkso-3-appx:ln:l-nearest}
        \EndWhile
        \State \Return $S_\cov$
        \Ensure $S_\cov \subseteq F$ \Comment{open facilities}
    \end{algorithmic}
\end{algorithm}
We argue that \Cref{alg:ufkso-3-appx} opens at most $k$ facilities, and that if $N_\ell(j,F)$ is opened, then $\child(j)$ is served within distance $3r$. Formally,
\begin{lemma}\label{lma:ufkso-distance}
    Given $\paren{\set{\cov_v}_{v \in C}, \set{x_i}_{i \in F}}$ satisfying \eqref{lp:fkso-weak:m}-\eqref{lp:fkso-weak:bounds},
    \begin{itemize}
        \item $\abs{S_\cov} \leq k$, and
        \item Let $R'_\cov$ be the clients $j$ for which $N_\ell(j,F)$ was added to $S_\cov$ in Line~\ref{alg:ufkso-3-appx:ln:l-nearest}. Then $\forall j \in R'_\cov$, $\forall v \in \child(j)$, $d_\ell(v,S_\cov) \leq 3r$.
    \end{itemize}
\end{lemma}
\begin{proof}
    \begin{itemize}
        \item Line~\ref{alg:ufkso-3-appx:ln:l-nearest} adds $\ell$ facilities to $S_\cov$ in each iteration. So Line~\ref{alg:ufkso-3-appx:ln:s-loop} ensures that $\abs{S_\cov} \leq k$.
        \item Consider $v \in \child(j)$, $j \in R'_\cov$. By triangle inequalities, $d_\ell(v,S_\cov) \leq d(v,j) + d_\ell(j,S_\cov)$. By Line~\ref{alg:ufkso-3-appx:ln:2opt-ball}, $d(v,j) \leq 2r$. Since $N_\ell(j,F) \subseteq S_\cov$, $d_\ell(j,S_\cov) \leq d_\ell(j,F)$; and by Line~\ref{alg:ufkso-3-appx:ln:positive-cov}, $\cov_j > 0$, i.e. by \eqref{lp:fkso-weak:inliers}, $d_\ell(j,F) \leq r$.
    \end{itemize}
\end{proof}

It remains to show that $\sum_{j \in R'_\cov}\abs{\child(j)} \geq m$. We have
\begin{align}
\sum_{j \in R_\cov}\abs{\child(j)}\cov_j \geq \sum_{v \in C : \cov_v > 0}\cov_v = \sum_{v \in C}\cov_v \geq m\,,\label{eq:wtd-avg}
\end{align}
where the first inequality is by Line~\ref{alg:ufkso-3-appx:ln:max-cov} and the last inequality is by \eqref{lp:fkso-weak:m}. We also have
\begin{align}
    \sum_{j \in R_\cov}\cov_j \leq \sum_{j \in R_\cov}\sum_{i \in F \cap B(j,r)}\frac{x_i}\ell \leq \sum_{i \in F}\frac{x_i}\ell \leq \frac k \ell\,,\label{eq:sum-of-wts}
\end{align}
where the first inequality is by \eqref{lp:fkso-weak:opt-ball}; the second inequality is because $R_\cov$ is well-separated; and the last inequality is by \eqref{lp:fkso-weak:k}. So we can view the LHS in \eqref{eq:wtd-avg} as a weighted sum of $\abs{\child(j)}$ values, the weights being $\cov_j$'s. Since this weighted sum is $\geq m$ and the weights sum to $\leq k/\ell$, the $\floor{\frac k \ell}$ largest child-sets must contain at least $\frac m {k/\ell} \cdot \floor{\frac k \ell}$ elements. Hence, if $\ell \mid k$, we are done.

In fact, we observe the following even when $\ell \nmid k$: if we can replace the RHS in \eqref{eq:sum-of-wts} with $\floor{\frac k \ell}$, then the weighted-sum argument would yield $\frac m {\floor{k/\ell}}\cdot\floor{\frac k \ell} = m$. To achieve this, observe that the argument in \eqref{eq:sum-of-wts} applies to any well-separated set $R \subseteq C$, yielding $\sum_{j \in R}\cov_j \leq k/\ell$. Also, for any integral solution, the RHS can be replaced by its floor. Thus the following are valid inequalities:
\begin{align}
    \forall R \subseteq C: R\text{ is well-separated},\quad \sum_{j \in R}\cov_j \leq \floor{\frac k \ell}\,.\tag{WLCut}\label{lp:fkso-weak:floor}
\end{align}
We have showed that if \eqref{lp:fkso-weak:floor} holds for $R = R_\cov$ then we are done, i.e.
\begin{lemma}\label{lma:ufkso-inliers}
    Given $\paren{\set{\cov_v}_{v \in C}, \set{x_i}_{i \in F}}$ satisfying \eqref{lp:fkso-weak:m}-\eqref{lp:fkso-weak:bounds}, if \eqref{lp:fkso-weak:floor} holds for $R = R_\cov$ where $R_\cov$ is constructed as per Lines~\ref{alg:ufkso-3-appx:ln:filtering-start}-\ref{alg:ufkso-3-appx:ln:filtering-end}, then $S_\cov$ is a $3$-approximation.
\end{lemma}

Using this, we now present our overall algorithm via a round-or-cut schema.
\begin{proof}[Proof of \Cref{thm:ufkso-3-appx}]
    Given $\paren{\set{\cov_v}_{v \in C}, \set{x_i}_{i \in F}}$ satisfying \eqref{lp:fkso-weak:m}-\eqref{lp:fkso-weak:bounds}, we round as per Lines~\ref{alg:ufkso-3-appx:ln:filtering-start}-\ref{alg:ufkso-3-appx:ln:filtering-end} to obtain $R_\cov \subseteq C$. If \eqref{lp:fkso-weak:floor} holds for $R = R_\cov$, then we continue \Cref{alg:ufkso-3-appx} to obtain $S_\cov$ that satisfies the desired guarantees via \Cref{lma:ufkso-distance,lma:ufkso-inliers}. Otherwise, we know that the valid inequality \eqref{lp:fkso-weak:floor} for $R = R_\cov$ is violated. So we pass it to the ellipsoid algorithm as a separating hyperplane, obtaining fresh $\cov_v$'s with which we restart \Cref{alg:ufkso-3-appx}. By the guarantees of the ellipsoid algorithm, in polynomial time, we either round to get $S_\cov$, or detect that the guess of $r$ is too small.
\end{proof}

We conclude this section by exhibiting that the above algorithm fails for the general problem. In particular, we exhibit an infinite integrality gap when there are just two different fault-tolerances in the instance.
\subsection{\texorpdfstring{Gap example for \fkso}{Gap example for FkSO}}\label{sec:weak-lp-gap}
Consider \eqref{lp:fkso-weak:opt-ball} generalized to \fkso:
\begin{align}
    \forall v \in C,\quad &\sum_{i \in F \cap B(v,r)}x_{iv} \geq \ell_v\cov_v\,;\label{lp:fkso-weak:opt-ball'}\tag{WL3$'$}
\end{align}
and a similar generalization of \eqref{lp:fkso-weak:floor}:
\begin{align}
    \forall R \subseteq C : R\text{ is well-separated},\quad &\ceil{\sum_{v \in R}\ell_v\cov_v} \leq k\,.\label{lp:fkso-weak:floor'}\tag{WLCut$'$}
\end{align}
These, along with \eqref{lp:fkso-weak:m}-\eqref{lp:fkso-weak:k} and \eqref{lp:fkso-weak:inliers}-\eqref{lp:fkso-weak:bounds}, generalize the earlier LP to \fkso. We now show an infinite integrality gap w.r.t. this LP.
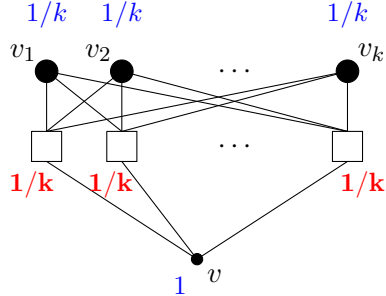
\begin{figure}[ht]
	\centering
	\begin{tikzpicture}
\filldraw (0,0) node[anchor = south east] {$v_1$} circle (0.15cm);
\filldraw (1,0) node[anchor = south east] {$v_2$} circle (0.15cm);
\draw (2.5,0) node {$\mathbf{\cdots}$};
\filldraw (4,0) node[anchor = south west] {$v_k$} circle (0.15cm);

\draw (-0.2,-0.8) -- (0.2,-0.8) -- (0.2,-1.2) -- (-0.2,-1.2) -- (-0.2,-0.8);
\draw (0.8,-0.8) -- (1.2,-0.8) -- (1.2,-1.2) -- (0.8,-1.2) -- (0.8,-0.8);
\draw (2.5,-1) node {$\mathbf{\cdots}$};
\draw (3.8,-0.8) -- (4.2,-0.8) -- (4.2,-1.2) -- (3.8,-1.2) -- (3.8,-0.8);

\draw (0,0) -- (0,-0.8);
\draw (0,0) -- (1,-0.8);
\draw (0,0) -- (4,-0.8);
\draw (1,0) -- (0,-0.8);
\draw (1,0) -- (1,-0.8);
\draw (1,0) -- (4,-0.8);
\draw (4,0) -- (0,-0.8);
\draw (4,0) -- (1,-0.8);
\draw (4,0) -- (4,-0.8);

\filldraw (2,-2.5) node[anchor = north west] {$v$} circle (0.075cm);

\draw (2,-2.5) -- (0,-1.2);
\draw (2,-2.5) -- (1,-1.2);
\draw (2,-2.5) -- (4,-1.2);

\draw (0,0.45) node[anchor = south] {{\color{blue}\small $1/k$}};
\draw (1,0.45) node[anchor = south] {{\color{blue}\small $1/k$}};
\draw (4,0.45) node[anchor = south] {{\color{blue}\small $1/k$}};
\draw (2,-2.6) node[anchor = north east]  {{\color{blue}\small $1$}};

\draw (-0.2,-1.2) node[anchor = north]  {{\color{red}\small $\bone/\bk$}};
\draw (0.85,-1.2) node[anchor = north]  {{\color{red}\small $\bone/\bk$}};
\draw (4.2,-1.2) node[anchor = north]  {{\color{red}\small $\bone/\bk$}};
\end{tikzpicture}
	\caption{\em \small One of the $k$ identical gadgets in the gap example, showing LP values in {\color{red}\textbf{red }($x$ values)} and {\color{blue}blue ($\cov$ values)}. The ``edges'' represent distance $1$, and all other distances are determined by making triangle inequalities tight. The fault-tolerances are $\ell_{v_1} = \ell_{v_2} = \cdots = \ell_{v_k} = k$, and $\ell_v = 1$.\label{fig:inf-gap}}
\end{figure}

Consider $k$ identical gadgets, each like in \Cref{fig:inf-gap}, infinitely apart from each other. Let $m = 2k$. The small client in each gadget ($v$ in \Cref{fig:inf-gap}) has fault-tolerance $1$. The big clients in each gadget ($v_1,v_2,\dots,v_k$ in \Cref{fig:inf-gap}) have fault-tolerance $k$. Within a gadget, an integral solution only benefits from either picking one facility to serve just the small client, or picking all facilities to serve all $(k+1)$ clients. So over all gadgets, an integral solution can either pick one facility per gadget, or pick all facilities in exactly one gadget, either way serving $k < m$ clients. Since all facilities are within distance $1$ of the clients in their gadget, the above is true for an integral solution with any radius dilation $\alpha \geq 1$.

But the LP can assign $x_i = 1/k$ to each of the $k^2$ facilities in the instance. This allows it to assign $\cov_v = 1$ to all the small clients, and $\cov_{v_1} = \cov_{v_2} = \cdots \cov_{v_k} = 1/k$ to all the big clients, thus serving $k \cdot 1 + k^2 \cdot \frac 1 k = 2k = m$ clients.

\section{\texorpdfstring{\fksofull}{Fault-tolerant k-Supplier with Outliers}}\label{sec:fkso}

In this section, we address \fkso in its full generality. We use $t$ to denote the number of distinct fault-tolerances in the instance, i.e. $\abs{\set{\ell_v : v \in C}} = t$. We prove that
\begin{theorem}\label{thm:fkso-O(t)}
    The \fkso problem admits a $\paren{\min\set{4t-1, 2^t+1}}$-approximation.
\end{theorem}

\subsection{Strong LP Relaxation and the Round-or-Cut Schema}\label{sec:stronglp}
To circumvent the gap example in \Cref{sec:ufkso}, we adapt the following stronger linear program idea from Chakrabarty and Negahbani~\cite{ChakrN2019}. 
As before, $r$ is the guess of the optimal solution, and we have the same fractional variables $\cov_v$ indicating coverage. However, we assert that these $\cov_v$'s arise as a convex combination of integral solutions. 
More precisely, we have {\em exponentially many} auxilary variables $\set{z_S}_{S \subseteq F : \abs S \leq k}$ indicating possible locations of open facilities and the fractional amount to which they are open.
When such a solution is opened, a client $v$ is ``covered'' if there are $\ell_v$ facilities in an $r$-neighborhood. To this end, for a client $v$, we define
the collection $\calF_v := \set{S \subseteq F : \abs S \leq k \land \abs{S \cap B(v,r)} \geq \ell_v}$ of solutions which can serve $v$. 
Therefore, the coverage $\cov_v$ is simply the total fractional weight of sets in $\calF_v$. Formally, if $r$ is a correct guess, then the following (huge) LP
has a feasible solution.

\begin{align}
    \sum_{v \in C}\cov_v &\geq m\label{lp:fkso:m}\tag{L1}\\
    \forall v \in C,\, \cov_v &=  \sum_{S \in \calF_v}z_S\label{lp:fkso:sanity}\tag{L2}\\
    \sum_{S \subseteq F:\abs S \leq k}z_S &\leq 1\label{lp:fkso:sum}\tag{L3}\\
    \forall S \subseteq F, \forall v \in C,\,0 \leq z_S, \cov_v &\leq 1\label{lp:fkso:bounds}\tag{L4}
\end{align}

\eqref{lp:fkso:m} enforces that at least $m$ clients must be covered. \eqref{lp:fkso:sanity} connects the $\cov_v$ and $z_S$ variables, ensuring that a client $v$ can only be covered via solutions in $\calF_v$. \eqref{lp:fkso:sum}-\eqref{lp:fkso:bounds} enforce convexity. \eqref{lp:fkso:bounds} also enforces that each client can be covered at most once.

\begin{observation}\label{obs:integer-hull}
    All $\cov_v$'s that satisfy \eqref{lp:fkso:m}-\eqref{lp:fkso:bounds} also satisfy \eqref{lp:fkso-weak:m}-\eqref{lp:fkso-weak:bounds}.
\end{observation}

Also observe that we cannot efficiently figure out whether the above system is feasible or not; indeed, if so we would solve the \fksofull problem optimally. 
Nevertheless, one can use the round-or-cut schema to obtain an {\em approximation} algorithm. In order to do so, the first step is to use the {\em dual} of the 
above system to obtain the collection of all valid inequalities on the $\cov_v$'s. Recall, a valid inequality is one that every feasible $\cov_v$ must satisfy; the lemma below from the literature~\cite{ChakrGK2020}, in some sense, eliminates all the $z_S$ variables from the above program.
\begin{lemma}[{\cite[Lemma 10]{ChakrN2019}}]\label{lma:lambdas}
    Given real numbers $\set{\lambda_v}_{v \in C}$ such that
    \begin{align}
        \forall S \subseteq F,\quad \sum_{v \in C: S \in \calF_v}\lambda_v < m\,,\label{eq:lambda-lemma:1}\tag{$\lambda 1$}
    \end{align}
    the following is a valid inequality for \eqref{lp:fkso:m}-\eqref{lp:fkso:bounds}:
    \begin{align}
        \sum_{v \in C}\lambda_v\cov_v < m\,.\label{eq:lambda-lemma:2}\tag{$\lambda 2$}
    \end{align}
\end{lemma}
Given $\set{\lambda_v}_{v \in C}$, one cannot easily check \eqref{eq:lambda-lemma:1}, and thus, a priori, one cannot see the usefulness of the above lemma.
We now briefly describe its usefulness to the round-or-cut schema. The algorithm begins with values of $\set{0 \leq \cov_v \leq 1}_{v \in C}$ that satisfy \eqref{lp:fkso:m} -- such $\cov$ is straightforward to find. We then try to use these $\cov_v$'s to ``round'' and obtain a solution 
where clients are covered within distance $\alpha\cdot r$ for desired factor $\alpha$, and if we fail, then we find a valid inequality that ``cuts'' $\cov_v$ away from the above system. If we can do so, then we can feed this separating hyperplane to the ellipsoid algorithm which would give us new $\cov_v$'s. Repeating the above procedure a polynomial number of times, we would either obtain an $\alpha$-approximation, or prove that the above system is empty implying our guess $r$ was too small.
For \fkso, the ``round'' step is via the abstract concept of a ``good partition'' where the ``radius'' of the partition dictates the approximation factor; this definition and resulting rounding algorithm is described in~\Cref{sec:fkso-rounding}. For the ``cut'' step, we show that if our rounding algorithm fails, then we can use this failure to generate $\set{\lambda_v}_{v \in C}$'s that satisfy \eqref{eq:lambda-lemma:1} but not \eqref{eq:lambda-lemma:2}, leveraging our definition of ``good partitions''. This gives our separating hyperplane using~\Cref{lma:lambdas}, and we succeed in cutting, and thus we can run the round-or-cut schema. Subsequently, we construct good partitions. In~\Cref{sec:obtain-good-partition}, we describe two methods to do this: one with ``radius'' $(4t-1)$ and the other with radius $(2^t+1)$. In~\Cref{appsec:partition-gap}, we show a limitation of our approach, via an example where this ``radius'' can be $\Omega(t)$.

Before proceeding, we make one simplification: at the beginning of every rounding step, we discard any clients that have $\cov_v = 0$, and hereafter assume, without loss of generality, that $\forall v \in C$, $\cov_v > 0$.

\subsection{Good Partitions and Implementing Round-or-cut}\label{sec:fkso-rounding}

Given $\cov_v$'s for every $v\in C$, we define a notion of a ``good partition''. Before formally defining it, we explain this operationally, hopefully giving intuition for the definition.
We start with a finer partition, and the good partition $\calP$ coarsens it. As in previous algorithms discussed so far, we have $R\subseteq C$, a set of {\em representatives}. The finer partition is $\set{\child(j)}_{j \in R}$, as motivated by our algorithms for \fks in \Cref{sec:prelim} and \ufkso in \Cref{sec:ufkso}. This time, however, we want favorable properties from both of those algorithms to coincide -- we want, for $j\in R$ and $v \in \child(j)$, $\cov_v \leq \cov_j$ as well as $\ell_v \leq \ell_j$. These desired properties of the finer partition are formalized as Property~\ref{def:good-partition:refinement}.

The property above breaks the ``well-separated'' property of $R$, which was crucial in our other algorithms in \Cref{sec:prelim,sec:ufkso}. 
Therefore, instead of requiring $R$ to be well-separated, we coalesce the child-sets of certain representatives, to get a {\em coarsening} $\calP$
of $\set{\child(j)}_{j \in R}$ such that representatives {\em across} different parts of $\calP$ are indeed well-separated. This is Property~\ref{def:good-partition:well-sep}.

Our approximation ratio is then determined by the diameter of the parts $P$'s in the good partition; so we impose a radius bound on each $P \in \calP$, requiring that the highest-fault-tolerance client in each $P$ be not too far from the rest of $P$. This is Property~\ref{def:good-partition:radius}. We are now ready to present the formal definition.
\newpage
\begin{definition}[$(\rho,\cov)$-good partition]\label{def:good-partition}
    Given a parameter $\rho \in \RR$, and $\set{0 \leq \cov_v \leq 1}_{v \in C}$ satisfying \eqref{lp:fkso:m}, a partition $\calP$ of $C$ is $(\rho, \cov)$-good if there exists $R \subseteq C$ such that the following hold.
    \begin{enumerate}
        \item Every $v \in C$ is assigned to be the $\child$ of a $j \in R$, forming a partition $\set{\child(j)}_{j \in R}$ of $C$ that refines $\calP$. Also, $\forall j \in R, \forall v \in \child(j)$, $\cov_j \geq \cov_v$ and $\ell_j \geq \ell_v$. \label{def:good-partition:refinement}
        \item For any two $j,j' \in R$ that lie in different parts of $\calP$, $d(j,j') > 2r$.
        \label{def:good-partition:well-sep}
        \item For each $P \in \calP$, let $j_P := \argmax_{v \in P}\ell_v$ (breaking ties arbitrarily). Then $\forall v \in P$, $d(j_P,v) \leq \rho r$.\label{def:good-partition:radius}
    \end{enumerate}
\end{definition}

\begin{figure}[hbt]
	\centering
	\begin{tikzpicture}
    \draw[very thick] (0,0) ellipse (1 and 1.75);
    \draw[very thick] (3,0) ellipse (1 and 2);
    \draw[very thick] (6,0) ellipse (1 and 1.5);

    \draw[thick] (2.05,0.7) -- (3.95,0.7);
    \draw[thick] (2.05,-0.7) -- (3.95,-0.7);
    \draw[thick] (-1,0) -- (1,0);

    \filldraw (0,0.5) circle (0.15);
        \draw (0,0.5) circle (0.22);
    \draw (0.7,0.7) circle (0.15);
    \draw (-0.5,1) circle (0.07);
    \filldraw (0,-0.5) circle (0.15);
    \draw (0.2,-1) circle (0.07);
    \filldraw (3,1) circle (0.07);
    \filldraw (3,0) circle (0.25);
        \draw (3,0) circle (0.33);
    \draw (2.5,-0.5) circle (0.15);
    \draw (3.5,0.5) circle (0.07);
    \filldraw (3,-1) circle (0.15);
    \draw (2.7,-1.5) circle (0.07);
    \draw (3.2,-1.5) circle (0.15);
    \filldraw (6,0) circle (0.07);
        \draw (6,0) circle (0.13);
    \draw (6.8,0) circle (0.07);

    \draw (0,0.5) -- (0.55,0.7);
    \draw (0,0.5) -- (-0.45,0.95);
    \draw (-0.45,0.95) .. controls (0,1) and (0.2,1) .. (0.55,0.7);
    \draw (0,-0.5) -- (0.15,-0.95);
    \draw (0,0.5) -- (0,-0.5);
    \draw (3,0) .. controls (2.5,0.2) and (2.5,0.8) .. (3,1);
    \draw (3,1) -- (3.44,0.5);
    \draw (3,0) -- (2.6,-0.4);
    \draw (3,0) -- (3.44,0.44);
    \draw (2.6,-0.6) -- (3,-1);
    \draw (3,-1) -- (2.75,-1.45);
    \draw (3,-1) -- (3.2,-1.35);
    \draw (6,0) -- (6.73,0);

    \draw (-0.45,0.95) .. controls (1,2) and (2,2) .. (3,1);
    \draw (0.27,-1) -- (2.35,-0.5);
    \draw (6,0) -- (3.35,-1.5);
    \draw (0.27,-1) .. controls (2,-3)  and (5,-3) .. (6,0);
\end{tikzpicture}
	\caption{\em \small An example of a $(6,\cov)$-good partition $\calP$ (\Cref{def:good-partition}). The ellipses represent $\calP$, and their subdivisions represent the $\child$ sets. All the circles are clients, with the filled-in circles being $R$, and among those, the double borders indicate the $j_P$'s. $\cov$ values are $1$ on $R$ and $1/2$ elsewhere. $\ell_v$ values are indicated by the sizes of the circles. The ``edges'' represent distance $2r$, and all other distances are obtained by making triangle inequalities tight.}
\end{figure}
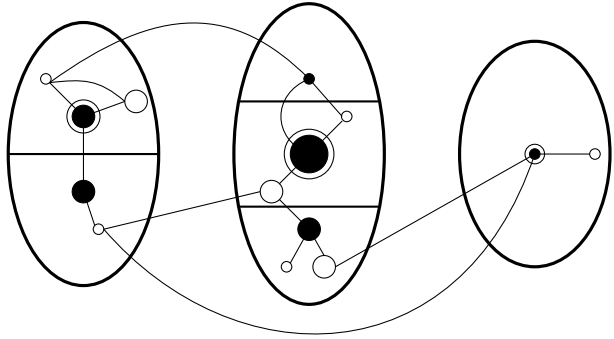

We observe here that the child-sets constructed in \Cref{sec:ufkso} are themselves a good partition, so for \ufkso, we did not need to coarsen it. This will not necessarily be the case for child-sets that we construct in \Cref{sec:obtain-good-partition}. We also observe that
\begin{observation}\label{obs:jP-in-R}
    In a $(\rho,\cov)$-good partition $\calP$, by Property~\ref{def:good-partition:refinement}, the $j_P$'s in Property~\ref{def:good-partition:radius} can be chosen such that $\forall P \in \calP$, $j_P \in R$. So we can assume, without loss of generality, that all $j_P$'s are in $R$.
\end{observation}
We prove that a good partition suffices to achieve our desired approximation. That is,
\begin{theorem}\label{thm:rc-for-good-partition}
    If we have a feasible instance with a $(\rho,\cov)$-good partition, then in polynomial time, we can either obtain a $(\rho+1)$-approximation, or identify a valid inequality for \eqref{lp:fkso:m}-\eqref{lp:fkso:bounds} that is violated by $\cov$.
\end{theorem}

To prove \Cref{thm:rc-for-good-partition}, we solve a budgeting problem on the $(\rho,\cov)$-good partition. We want to distribute our budget of $k$ facilities among the $P \in \calP$, assigning each $P \in \calP$ with $k_P$ facilities that are within distance $(\rho+1)r$ of the clients in $P$. Here $k_P$ must be at most $\ell_P := \ell_{j_P}$, because at most $\ell_P$ facilities are guaranteed to exist within a bounded distance of clients in $P$. The payoff from assigning $k_P$ facilities to $P$ in this way is that the clients $\set{v \in P:\ell_v \leq k_P}$ are served within distance $(\rho+1)r$. So if $\sum_{P \in \calP}\abs{\set{v \in P : \ell_v \leq k_P}} \geq m$, we have our $(\rho+1)$-approximation. Therefore, we want our choice of $k_P$'s to maximize $\sum_{P \in \calP}\abs{\set{v \in P : \ell_v \leq k_P}}$, and this maximum to be $\geq m$. However, our analysis can only handle clients from well-separated sets; so instead, we maximize the following lower-bound on our desired quantity: $\sum_{P \in \calP}\sum_{j \in R \cap P : \ell_j \leq k_P}\abs{\child(j)}$, where we under-count by only considering $v \in \child(j)$ served if $j$ is served. Formally, our budgeting problem is the following.
\begin{definition}[Budgeting over a $(\rho,\cov)$-good partition]\label{def:budgeting}
    Given a $(\rho,\cov)$-good partition $\calP$, let $\ell_P := \max_{v \in P}\ell_v$. Find $\set{k_P \leq \ell_P}_{P \in \calP}$ such that $\sum_{P \in \calP}k_P \leq k$, maximizing $\sum_{P \in \calP}\sum_{j \in R \cap P : \ell_j \leq k_P}\abs{\child(j)}$.
    Let $\opt_B(\calP)$ denote this maximum.
\end{definition}
In \Cref{lma:round}, we show that if $\opt_B(\calP) \geq m$, then we can round. Then in \Cref{lma:cut}, we see that if $\opt_B(\calP) < m$, then we can cut. \Cref{lma:dp} shows that $\opt_B(\calP)$ can be found efficiently. Together, these three lemmas yield the proof of \Cref{thm:rc-for-good-partition}.
\begin{lemma}\label{lma:round}
    Given a $(\rho,\cov)$-good partition $\calP$, if $\opt_B(\calP) \geq m$, then we have a $(\rho+1)$-approximation.
\end{lemma}
\begin{proof}
    Let $\set{k_P}_{P \in \calP}$ be an optimal solution to the budgeting problem (\Cref{def:budgeting}). Define $S := \cup_{P \in \calP}N_{k_P}(j_P,F)$. So $\abs S \leq k$. We show that $S$ serves $\geq m$ clients within distance $(\rho+1)r$.
    
    Define $T := \uplus_{P \in \calP}\uplus_{j \in R \cap P: \ell_j \leq k_P}\child(j)$. Then $\abs T = \sum_{P \in \calP}\sum_{j \in R \cap P: \ell_j \leq k_P}\abs{\child(j)} = \opt_B(\calP) \geq m$. We complete this proof by showing that $\forall v \in T$, $d_{\ell_v}(v,S) \leq (\rho+1)r$. For this, fix $v \in T$. By triangle inequalities, we have that $d_{\ell_v}(v,S) \leq d(v,j_P) + d_{\ell_v}(j_P,S)$. By Property~\ref{def:good-partition:radius}, $d(v,j_P) \leq \rho r$, so it remains to show that $d_{\ell_v}(j_P,S) \leq r$.
    
    By definition of $T$, $d_{\ell_v}(j_P,S) \leq d_{k_P}(j_P,S)$. Since $N_{k_P}(j_P,F) \subseteq S$, $d_{k_P}(j_P,S) \leq d_{k_P}(j_P,F)$. By definitions of $k_P$ and $\ell_P$, $d_{k_P}(j_P,F) \leq d_{\ell_P}(j_P,F) = d_{\ell_{j_P}}(j_P,F)$. But $\cov_{j_P} > 0$; so by \Cref{obs:integer-hull} and \eqref{lp:fkso-weak:inliers}, $d_{\ell_{j_P}}(j_P,F) \leq r$.
\end{proof}
\begin{lemma}\label{lma:cut}
    Given a $(\rho,\cov)$-good partition $\calP$, if $\opt_B(\calP) < m$, then we find a valid inequality for \eqref{lp:fkso:m}-\eqref{lp:fkso:bounds} that is violated by $\cov$.
\end{lemma}
\begin{proof}We appeal to \Cref{lma:lambdas} mentioned in~\Cref{sec:stronglp}. $\forall v \in C$, define
\[\lambda_v := \begin{cases}
\abs{\child(v)} &\text{if }v \in R\text{, and}\\
0 &\text{otherwise.}\end{cases}\]
Note that
\begin{align*}
    \sum_{v \in C}\lambda_v\cov_v = \sum_{j \in R}\lambda_j\cov_j &= \sum_{j \in R}\abs{\child(j)}\cov_j = \sum_{j \in R}\sum_{v \in \child(j)}\cov_j\\
    &\geq \sum_{j \in R}\sum_{v \in \child(j)}\cov_v&\dots\text{by Property~\ref{def:good-partition:refinement}}\\
    &= \sum_{v \in C}\cov_v&\dots\text{by \Cref{def:good-partition}}\\
    &\geq m\,,&\dots\text{by \eqref{lp:fkso:m}}
\end{align*}
i.e. these $\lambda_v$'s violate \eqref{eq:lambda-lemma:2}. So by \Cref{lma:lambdas}, it suffices to show that \eqref{eq:lambda-lemma:1} holds for these $\lambda_v$'s. 

Suppose not, i.e. $\exists S_0 \subseteq F : \abs{S_0} \leq k$ and $\sum_{v \in C: S_0 \in \calF_v}\lambda_v \geq m$. Then, devise a candidate solution $\set{k'_P}_{P \in \calP}$ for the budgeting problem in \Cref{def:budgeting}, as follows. For each $P \in \calP$, if $\exists j \in R \cap P$ such that $S_0 \in \calF_j$, then set $k'_P$ to be the largest fault-tolerance among such $j$'s; that is, where $j'_P := \argmax_{j \in R \cap P: S_0 \in \calF_j}\ell_j$, set $k'_P := \ell_{j'_P}$. Otherwise, i.e. when there is no such $j$ and $j'_P$ is not well-defined, set $k'_P := 0$. By definitions, $\forall P \in \calP$, $k'_P \leq \ell_P$.

Also, by Property~\ref{def:good-partition:well-sep}, $\set{B(j'_P,r)}_{P \in \calP}$ is pairwise disjoint. Since $S_0 \in \calF_{j'_P}$ for each $P \in \calP$, we then have $\sum_{P \in \calP}k'_P \leq \sum_{P \in \calP}\abs{S_0 \cap B(j'_P,r)} \leq \abs{S_0} \leq k$. So $\set{k'_P}_{P \in \calP}$ is indeed a candidate solution for the budgeting problem. We evaluate the objective function of the budgeting problem (see \Cref{def:budgeting}) on $\set{k'_P}_{P \in \calP}$:
\begin{align*}
    \sum_{P \in \calP}\sum_{j \in R \cap P : \ell_j \leq k'_P}\abs{\child(j)} &= \sum_{P \in \calP}\sum_{j \in R \cap P: \ell_j \leq k'_P}\lambda_j\\
    &\geq \sum_{P \in \calP}\sum_{j \in R \cap P : S_0 \in \calF_j}\lambda_j&\dots\text{by choice of $k'_P$'s}\\
    &=\sum_{j \in R:S_0 \in \calF_j}\lambda_j &\dots\text{by \Cref{def:good-partition}}\\
    &=\sum_{v \in C: S_0 \in \calF_v}\lambda_v&\dots\text{by choice of $\lambda_v$'s}\\
    &\geq m &\dots\text{by supposition.}
\end{align*}
So $\set{k'_P}_{P \in \calP}$ is a candidate solution to the budgeting problem, for which the objective function evaluates to $\geq m$, contradicting $\opt_B(\calP) < m$. Hence \eqref{eq:lambda-lemma:1} holds for our chosen $\lambda_v$'s, and \eqref{eq:lambda-lemma:2} is the desired valid inequality that is violated by $\cov$.
\end{proof}

\begin{lemma}\label{lma:dp}
    The budgeting problem in \Cref{def:budgeting} can be solved in polynomial time.
\end{lemma}
\begin{proof}
    We proceed via dynamic programming. Let $N := \abs{\calP}$. Without loss of generality, say $\calP =: \set{P_1,P_2,\dots,P_N}$. For brevity, $\forall a \in [N]$, we say $L_a := \ell_{P_a}$. To handle base cases in our DP, we set the convention that $P_0 := \emptyset$. Now define the entries in our DP table: $\forall \nu \in [N]\cup\set 0$ and $\forall b \in [k]\cup\set 0$,
    \begin{align}
        M[\nu,b] := \max_{\set{k_a \leq L_a}_{a=1}^\nu : \sum_{a=1}^\nu k_a \leq b}\sum_{a=1}^\nu\sum_{j \in R \cap P_a : \ell_j \leq k_a}\abs{\child(j)}\,.\label{eq:DP-defn}\tag{DP-defn}
    \end{align}
    The desired entry is $M[N,k]$, as the corresponding $\set{k_a}_{a=1}^N$ becomes, upon renaming as $\set{k_{P_a} = k_a}_{a=1}^N$, the $k_P$'s that we want.
    
    The base cases are: $M[0,0] = 0$; $\forall \nu \in [N]$, $M[\nu,0] = 0$; and $\forall b \in [k]$, $M[0,b] = 0$. The DP table has $O(Nk) = O(nk)$ entries; so in polynomial time, we can fil it via the following recurrence.
    \begin{align}
        M[\nu,b] := \max_{\ell = 0}^{\min(b,L_\nu)}\paren{M[\nu-1,b-\ell] + \sum_{j \in R \cap P_\nu: \ell_j \leq \ell}\abs{\child(j)}}\,.\label{eq:DP-rec}\tag{DP-rec}
    \end{align}
    We also remember, for each entry $M[\nu,b]$, the $\ell$ that maximizes the RHS of \eqref{eq:DP-rec}. Note, in \eqref{eq:DP-defn}, that the RHS for $M[N,k]$ corresponds, up to renaming, with the RHS in the objective function (see \Cref{def:budgeting}). Thus it remains to show that \eqref{eq:DP-rec} is correct wrt \eqref{eq:DP-defn}.
    \begin{itemize}
        \item To show that LHS $\leq$ RHS, consider the solution $\set{k^*_a}_{a=1}^\nu$ corresponding to $M[\nu,b]$. By \eqref{eq:DP-defn}, $k^*_\nu \leq \min(b,L_\nu)$. So $\set{k^*_a}_{a=1}^{\nu-1}$ is a candidate solution for $M\bracket{\nu-1,b-k^*_a}$, i.e. $\sum_{a=1}^{\nu-1} \sum_{j \in R \cap P_a : \ell_j \leq k^*_a}\abs{\child(j)} \leq M\bracket{\nu-1,b-k^*_a}$, so
        \begin{align*}
            \text{LHS} = M[\nu,b] &= \sum_{a=1}^\nu\sum_{j \in R\cap P_a:\ell_j \leq k^*a}\abs{\child(j)}\\
            &\leq M\bracket{\nu-1,b-k^*_a} + \sum_{j \in R \cap P_\nu:\ell_j \leq k^*_\nu}\abs{\child(j)} \leq \text{RHS}
        \end{align*}
        since the RHS is a maximum.
        \item To show that RHS $\leq$ LHS, fix an $\ell \in \set{0,\dots,\min(b,L_\nu)}$, and let $\set{k'_a}_{a=1}^{\nu-1}$ be the solution corresponding to $M[\nu-1,b-\ell]$. Setting $k'_\nu = \ell$ yields $\set{k'_a}_{a=1}^\nu$, a candidate solution for $M[\nu,b]$. So
        \begin{align*}
            M[\nu-1,b-\ell] + \sum_{j \in R \cap P_\nu: \ell_j \leq \ell}\abs{\child(j)} &= \sum_{a=1}^\nu\sum_{j \in R \cap P_a:\ell_j \leq k'_a}\abs{\child(j)} \leq M[\nu,b] = \text{LHS}
        \end{align*}
        since $M[\nu,b]$ is a maximum by \eqref{eq:DP-defn}.
        
        As the RHS maximizes over $\ell \in \set{0,\dots,\min(b,L_\nu)}$, we are done.
    \end{itemize}
\end{proof}

\begin{proof}[Proof of \Cref{thm:rc-for-good-partition}]
    Given a $(\rho,\cov)$-good partition $\calP$, we solve the budgeting problem (\Cref{def:budgeting}), which we can do efficiently due to \Cref{lma:dp}, and obtain $\opt_B(\calP)$. If $\opt_B(\calP) \geq m$, \Cref{lma:round} guarantees a $(\rho+1)$-approximation; otherwise, \Cref{lma:cut} gives a valid inequality that is violated by $\cov$. We pass the valid inequality as a separating hyperplane to the ellipsoid algorithm, and restart our rounding process with fresh $\cov_v$'s. By the guarantees of ellipsoid, in polynomial time, we either round to obtain a $(\rho+1)$-approximation, or detect that the guess of $r$ is too small.
\end{proof}

\subsection{Obtaining a good partition} \label{sec:obtain-good-partition}
\begin{theorem} Given $\set{0 \leq \cov_v \leq 1}_{v \in C}$, in polynomial time, we can obtain the following:
\begin{enumerate}
    \item a $(4t-2,\cov)$-good partition, and
    \item a $(2^t,\cov)$-good partition.
\end{enumerate}
\end{theorem}
\Cref{thm:rc-for-good-partition} follows from \Cref{lma:partition:4t-1,lma:partition:exp-t}.
\begin{algorithm}[tbh]\caption{Finding a $(4t-2,\cov)$-good partition}\label{alg:partition:4t-1}
    \begin{algorithmic}[1]
        \Require $\set{0 \leq \cov_v \leq 1}_{v \in C}$
        \State $U \gets C$
        \State $R \gets \emptyset$
        \While{$U \neq \emptyset$}
            \State $j \gets \argmax_{v \in U}\cov_v$\label{alg:partition:4t-1:ln:max-cov-R}
            \State $R \gets R \cup \set j$
            \State $\child(j) \gets \set{v \in U: d(v,j) \leq 2tr \land \ell_v \leq \ell_j}$\label{alg:partition:4t-1:ln:child}
            \State $U \gets U \setminus \child(j)$ \label{alg:partition:4t-1:ln:remove-child}
        \EndWhile
        \State $\calP \gets \emptyset$
        \State $G \gets (R,E:=\set{\set{j,j'} : d(j,j') \leq 2r})$ \label{alg:partition:4t-1:ln:G}\Comment{undirected graph}
        \State $\calC \gets$ connected components of $G$\label{alg:partition:4t-1:ln:conn}
        \State $\calP \gets \set{\cup_{j \in V}\child(j)}_{V \in \calC}$\label{alg:partition:4t-1:P-from-conn}
        \Ensure A partition $\calP$ of $C$.
    \end{algorithmic}
\end{algorithm}
\begin{lemma}\label{lma:partition:4t-1}
    \Cref{alg:partition:4t-1} yields a $(4t-2,\cov)$-good partition.
\end{lemma}
\begin{proof}
    Consider $\calP$, the output of \Cref{alg:partition:4t-1}, and the $\child$ and $R$ constructed alongside. Line~\ref{alg:partition:4t-1:ln:remove-child} ensures that $\set{\child(j)}_{j \in R}$ is a partition of $C$. Line~\ref{alg:partition:4t-1:P-from-conn} ensures that this partition is a refinement of $\calP$. Lines~\ref{alg:partition:4t-1:ln:max-cov-R} and~\ref{alg:partition:4t-1:ln:child} construct $\child$ as desired, ensuring that $\forall j \in R$, $\forall v \in \child(j)$, $\cov_j \geq \cov_v$ and $\ell_j \geq \ell_v$. So Property~\ref{def:good-partition:refinement} holds.

    Now consider $P_1,P_2 \in \calP$, $x_1 \in R \cap P_1, x_2 \in R \cap P_2 : P_1 \neq P_2$. By Lines~\ref{alg:partition:4t-1:ln:G}-\ref{alg:partition:4t-1:ln:conn}, $R \cap P_1$ and $R \cap P_2$ are distinct connected components in $\calC$, so $\set{x_1,x_2} \notin E$, i.e. $d(x_1,x_2) > 2r$. This shows that Property~\ref{def:good-partition:well-sep} holds.

    Finally, consider $P \in \calP$, and $v \in P$ s.t. $v \in \child(j_1)$ for $j_1 \in R$. By Line~\ref{alg:partition:4t-1:P-from-conn}, $j_1 \in R \cap P$. Also consider a different $j_2 \in R \cap P$. By Lines~\ref{alg:partition:4t-1:ln:G}-\ref{alg:partition:4t-1:ln:conn}, $R \cap P \in \calC$. In $G$, consider $\pi$, the shortest $j_1$-$j_2$ path passing entirely through $R \cap P$. We claim that
    \begin{claim*}
        $\pi$ contains at most $t$ vertices.
    \end{claim*}
    \begin{claimproof}
        Suppose not. Then, by the pigeonhole principle, $\pi$ contains vertices $u,v \in R \cap P$ s.t. $u \neq v$ and $\ell_u = \ell_v$. Choose such $u,v$ minimizing $d(u,v)$, and consider the $u$-$v$ subpath $\pi'$ of $\pi$. If $\pi'$ contains $>t$ vertices, then we can replace $j_1,j_2$ with $u,v$ and repeat our argument to obtain a smaller $d(u,v)$ -- contradicting our choice of $u,v$. So $\pi'$ contains $\leq t$ vertices, i.e. $d(u,v) \leq 2(t-1)r$; but since $u,v \in R$, this contradicts Line~\ref{alg:partition:4t-1:ln:child}.
    \end{claimproof}
    So $d(j_1,j_2) \leq 2(t-1)r$, i.e. by Line~\ref{alg:partition:4t-1:ln:child}, $d(v,j_2) \leq d(v,j_1) + d(j_1,j_2) \leq 2tr + 2(t-1)r = (4t-2)r$. We have just showed that, $\forall v \in P, j \in R \cap P$, $d(v,j) \leq (4t-2)r$. By \Cref{obs:jP-in-R}, this implies Property~\ref{def:good-partition:radius} for $\rho = (4t-2)$.
\end{proof}

\begin{algorithm}[bht]\caption{Finding a $(2^t,\cov)$-good partition}\label{alg:partition:exp-t}
    \begin{algorithmic}[1]
        \State $U \gets C$
        \State $(R,E) \gets (\emptyset,\emptyset)$ \Comment{initializing an empty directed forest}
        \State $\forall v \in U$, $\height(v) \gets 0$ \Comment{height in the forest; $\height(v) = 0 \implies v \notin R$}
        \State $\Roots \gets \emptyset$ \Comment{tracking roots in the forest}
        \While{$U \neq \emptyset$}
            \State $j \gets \argmax_{v \in U}\cov_v$\label{alg:partition:exp-t:ln:max-cov}
            \State $R \gets R \cup \set j$
            \State $E \gets E \cup \set{(j,j') : j' \in \Roots \land d(j,j') \leq 2^{\height(j')}r}$\label{alg:partition:exp-t:ln:add-edges}
            \State $\Roots \gets \paren{\Roots \setminus \set{j' : (j,j') \in E}} \cup \set j$\label{alg:partition:exp-t:ln:update-roots}
            \State $\height(j) \gets 1 + \max_{(j,j') \in E}\height(j')$ \Comment{convention: max over $\emptyset$ is $0$}\label{alg:partition:exp-t:ln:calc-height}
            \State $\child(j) \gets \set{v \in U: d(v,j) \leq 2^{\height(j)}r \land \ell_v \leq \ell_j}$\label{alg:partition:exp-t:ln:child}
            \State $U \gets U \setminus \child(j)$\label{alg:partition:exp-t:ln:remove-child}
        \EndWhile
        \State $\calT \gets$ connected components in the forest $(R,E)$\label{alg:partition:exp-t:ln:trees}\Comment{each component induces a tree}
        \State $\calP \gets \set{\cup_{j \in V}\child(j)}_{V \in \calT}$\label{alg:partition:exp-t:ln:P-from-trees}
    \end{algorithmic}
\end{algorithm}

\begin{lemma}\label{lma:partition:exp-t}
    \Cref{alg:partition:exp-t} yields a $(2^t,\cov)$-good partition.
\end{lemma}
\begin{proof}
    Consider $\calP$, the output of \Cref{alg:partition:exp-t}, and the $\child$ and $R$ constructed alongside. Note that, since Line~\ref{alg:partition:exp-t:ln:add-edges} only creates edges to $\Roots$, and Line~\ref{alg:partition:exp-t:ln:update-roots} updates $\Roots$ accordingly, $(R,E)$ is indeed a forest.

    Line~\ref{alg:partition:exp-t:ln:remove-child} ensures that $\set{\child(j)}_{j \in R}$ is a partition of $C$. Line~\ref{alg:partition:exp-t:ln:P-from-trees} ensures that this partition is a refinement of $\calP$. Lines~\ref{alg:partition:exp-t:ln:max-cov} and~\ref{alg:partition:exp-t:ln:child} construct $\child$ as desired, ensuring that $\forall j \in R$, $\forall v \in \child(j)$, $\cov_j \geq \cov_v$ and $\ell_j \geq \ell_v$. So Property~\ref{def:good-partition:refinement} holds.

    Now consider $P_1,P_2 \in \calP$, $x_1 \in R\cap P_1$, $x_2 \in R \cap P_2$. Without loss of generality, suppose $x_2$ was added to $R$ after $x_1$; if $d(x_1,x_2) \leq 2r$, then by Lines~\ref{alg:partition:exp-t:ln:add-edges} and~\ref{alg:partition:exp-t:ln:calc-height}, we would have $d(x_2,x_1) \in E$, i.e. $x_1,x_2$ would lie in the same connected component in $\calT$. So by Lines~\ref{alg:partition:exp-t:ln:trees}-\ref{alg:partition:exp-t:ln:P-from-trees}, $P_1 = P_2$. This shows that Property~\ref{def:good-partition:well-sep} holds.

    Finally, note that
    \begin{claim}\label{clm:partition:exp-t:forest-edge}
        $(j,j') \in E \implies \ell_j > \ell_{j'}$.
    \end{claim}
    \begin{claimproof}
        Since $(j,j') \in E$, we know that $j'$ was added to $R$ before $j$, and $d(j,j') \leq 2^{\height(j')}r$. So if $\ell_j \leq \ell_{j'}$, then by Line~\ref{alg:partition:exp-t:ln:child}, we would have $j \in \child(j')$, contradicting the fact that $j \in R$.
    \end{claimproof}
    Now fix $P \in \calP$, and consider $j_P$ which, by \Cref{obs:jP-in-R}, lies in $R \cap P$, and hence by Lines~\ref{alg:partition:exp-t:ln:trees}-\ref{alg:partition:exp-t:ln:P-from-trees}, $R \cap P$ induces a tree in $(R,E)$. \Cref{clm:partition:exp-t:forest-edge} tells us that $j_P$ is the root in this tree, and that $\height(j_P) \leq t$. So by Line~\ref{alg:partition:exp-t:ln:add-edges}, for any $j \in R \cap P$, $d(j_P,j) \leq \paren{2^{\height(j_P)} - 2^{\height(j)}}r \leq \paren{2^t - 2^{\height(j)}}r$. Now consider $v \in P : v \in \child(j)$ for a $j \in R \cap P$. Then $d(v,j) \leq 2^{\height(j)}$, so $d(v,j_P) \leq d(v,j) + d(j,j_P) \leq \paren{2^t - 2^{\height(j)} + 2^{\height(j)}}r = 2^t r$. Thus Property~\ref{def:good-partition:radius} holds for $\rho = 2^t$.
\end{proof}

\section{Conclusion}

In this paper, we have studied the \fksofull problem and presented a $(4t-1)$-approximation when there are $t$ distinct fault tolerances. While this gives the optimal $3$-approximation for the uniform version of the problem (improving upon the recent result~\cite{InamdV2020}), the parameter $t$ could be as large as $k$.
To obtain our result, we needed to resort to the powerful hammer of the round-or-cut schema, and indeed used a very strong LP relaxation. This was necessary since, as we saw in \Cref{sec:weak-lp-gap}, natural LP relaxations and their strengthenings have unbounded integrality gaps. We also show a $\Omega(t)$-bottleneck to our approach (\Cref{appsec:partition-gap}), and this raises the intriguing question: are there $O(1)$-approximations for the \fkso problem? As noted in \Cref{sec:intro}, the authors are not aware of clustering problems where the version without outliers has a constant approximation (as we saw in \Cref{sec:prelim}, \fks does), but the outlier version doesn't. Perhaps \fkso is such a candidate example. This also raises the question of designing inapproximability results for metric clustering problems, which has not been explored much. We leave all these as interesting avenues of further study.


\bibliographystyle{plain}
\bibliography{references}

\appendix

\section{\texorpdfstring{Proof of \Cref{clm:fkso:weak-lp-validity}}{Proof of the weak LP's validity}}\label{appsec:omitted-proofs}
Consider a feasible solution $S^*$ that serves inliers $T^*$. Set
\begin{itemize}
    \item $\forall v \in C$, $\cov_v = \bone_{v \in T^*}$, and
    \item $\forall i \in F$, $y_i = \bone_{i \in S^*}$.
\end{itemize}
These satisfy \eqref{lp:fkso-weak:m}, \eqref{lp:fkso-weak:k}, and \eqref{lp:fkso-weak:bounds} by construction. Now note that, for a $v \in T^*$, $N_{\ell_v}(v,F) \subseteq S$; so \eqref{lp:fkso-weak:opt-ball} is satisfied. Furthermore, for a $v \in C$, if $d_{\ell_v}(v,F) > r$ then $v \notin T^*$, satisfying \eqref{lp:fkso-weak:inliers}.

Conversely, given an integral solution satisfying \eqref{lp:fkso-weak:m}-\eqref{lp:fkso-weak:bounds}, we can construct $S^* = \set{i \in F : y_i = 1}$, and $T^* = \set{v \in C : \cov_v = 1}$. \eqref{lp:fkso-weak:k} implies $\abs{S^*} \leq k$, and \eqref{lp:fkso-weak:m} implies $w(T^*) \geq W$. For any $v \in T^*$,
\begin{align*}
    \abs{S^* \cap B(v,r)} &= \sum_{i \in F \cap B(v,r)}y_i&\dots\text{by construction of }S^*\\
    &\geq \ell\cov_v = \ell\,, &\dots\text{by \eqref{lp:fkso-weak:opt-ball} and construction of }T^*
\end{align*}
so $d_{\ell_v}(v,S^*) \leq r$.

\section{Limiting Example for Good-Partition Rounding} \label{appsec:partition-gap}
In order to achieve a better approximation factor than $\Omega(t)$, we will need to move beyond the overall schema of using a good partition (\Cref{def:good-partition}) to round solutions to \eqref{lp:fkso:m}-\eqref{lp:fkso:bounds}. This can be seen via the following example, illustrated in \Cref{fig:part-gap}. Here $r = 1, n = t, m = 1$. $C$ is the set $\{v_1,\cdots,v_t\}$, with each client $v_a$ having fault-tolerance $\ell_{v_a}=a$. $F$ is the union of $t$ sets $\set{F_a}_{a=1}^t$, where $F_a = \set{i_{a1},i_{a2},\dots,i_{ak}}$, for a total of $tk$ facilities in $F$. Each client $v_a$ has distance $2$ to $v_{a+1}$ and $v_{a-1}$, and distance $1$ to each facility in $F_a$. Remaining distances are determined by making triangle inequalities tight.

Consider the following $(\cov,z)$ satisfying \eqref{lp:fkso:m}-\eqref{lp:fkso:bounds}. We set $z_{F_a} = \frac 1 {aH_t}$ for each $a \in [t]$, where $H_t$ is the $t$\textsuperscript{th} Harmonic number; and set all other $z_S$'s to zero. This allows us to set $\cov_{v_a} = \frac 1 {aH_t}$ for each $a \in [t]$. Under this $(\cov,z)$, observe that $\forall v_a,v_b \in C$, $v_a \neq v_a \land \cov_a \geq \cov_b \implies \ell_a < \ell_b$; so Property~\ref{def:good-partition:refinement} can only hold if all clients are in the same piece of the partition, i.e. $\calP = \set{C}$. This means that a $(\rho,\cov)$-good partition can only be attained for $\rho \geq 2(t-1)$, so upon applying \Cref{thm:rc-for-good-partition}, this approach attains a $(2t-1)$-approximation at best.

\begin{figure}[ht]
	\centering
	\begin{tikzpicture}[square/.style={regular polygon,regular polygon sides=4}]
\filldraw (0,0) node[anchor = south east] {$v_1$} circle (0.075cm);
\filldraw (2.5,0) node[anchor = south east] {$v_2$} circle (0.100cm);
\filldraw (5,0) node[anchor = south east] {$v_3$} circle (0.125cm);
\draw (7.5,0) node {$\mathbf{\cdots\quad\cdots\quad\cdots}$};
\filldraw (10,0) node[anchor = south west] {$v_t$} circle (0.150cm);

\draw (-0.9,-0.8) -- (-0.7,-0.8) -- (-0.7,-1) -- (-0.9,-1) -- (-0.9,-0.8);
\draw (-0.7,-0.8) node[anchor = south east] {\tiny $i_{11}$};
\draw (-0.3,-0.8) -- (-0.1,-0.8) -- (-0.1,-1) -- (-0.3,-1) -- (-0.3,-0.8);
\draw (-0.2,-0.8) node[anchor = south west] {\tiny $i_{12}$};
\draw (0.3,-0.9) node {$\cdots$};
\draw (0.6,-0.8) -- (0.8,-0.8) -- (0.8,-1) -- (0.6,-1) -- (0.6,-0.8);
\draw (0.6,-0.8) node[anchor = south west] {\tiny $i_{1k}$};
\draw (1.6,-0.8) -- (1.8,-0.8) -- (1.8,-1) -- (1.6,-1) -- (1.6,-0.8);
\draw (1.8,-0.8) node[anchor = south east] {\tiny $i_{21}$};
\draw (2.2,-0.8) -- (2.4,-0.8) -- (2.4,-1) -- (2.2,-1) -- (2.2,-0.8);
\draw (2.3,-0.8) node[anchor = south west] {\tiny $i_{22}$};
\draw (2.8,-0.9) node {$\cdots$};
\draw (3.1,-0.8) -- (3.3,-0.8) -- (3.3,-1) -- (3.1,-1) -- (3.1,-0.8);
\draw (3.1,-0.8) node[anchor = south west] {\tiny $i_{2k}$};
\draw (4.1,-0.8) -- (4.3,-0.8) -- (4.3,-1) -- (4.1,-1) -- (4.1,-0.8);
\draw (4.3,-0.8) node[anchor = south east] {\tiny $i_{31}$};
\draw (4.7,-0.8) -- (4.9,-0.8) -- (4.9,-1) -- (4.7,-1) -- (4.7,-0.8);
\draw (4.8,-0.8) node[anchor = south west] {\tiny $i_{32}$};
\draw (5.3,-0.9) node {$\cdots$};
\draw (5.6,-0.8) -- (5.8,-0.8) -- (5.8,-1) -- (5.6,-1) -- (5.6,-0.8);
\draw (5.7,-0.8) node[anchor = south west] {\tiny $i_{3k}$};
\draw (9.1,-0.8) -- (9.3,-0.8) -- (9.3,-1) -- (9.1,-1) -- (9.1,-0.8);
\draw (9.3,-0.8) node[anchor = south east] {\tiny $i_{11}$};
\draw (9.7,-0.8) -- (9.9,-0.8) -- (9.9,-1) -- (9.7,-1) -- (9.7,-0.8);
\draw (9.8,-0.8) node[anchor = south west] {\tiny $i_{12}$};
\draw (10.3,-0.9) node {$\cdots$};
\draw (10.6,-0.8) -- (10.8,-0.8) -- (10.8,-1) -- (10.6,-1) -- (10.6,-0.8);
\draw (10.6,-0.8) node[anchor = south west] {\tiny $i_{1k}$};

\draw[thin] (0,0) -- (-0.8,-0.8);
\draw[thin] (0,0) -- (-0.2,-0.8);
\draw[thin] (0,0) -- (0.7,-0.8);
\draw[thin] (2.5,0) -- (1.7,-0.8);
\draw[thin] (2.5,0) -- (2.3,-0.8);
\draw[thin] (2.5,0) -- (3.2,-0.8);
\draw[thin] (5,0) -- (4.2,-0.8);
\draw[thin] (5,0) -- (4.8,-0.8);
\draw[thin] (5,0) -- (5.7,-0.8);
\draw[thin] (10,0) -- (9.2,-0.8);
\draw[thin] (10,0) -- (9.8,-0.8);
\draw[thin] (10,0) -- (10.7,-0.8);

\draw[very thick] (0,0) -- (2.5,0);
\draw[very thick] (2.5,0) -- (5,0);
\draw[very thick] (5,0) -- (6,0);
\draw[very thick] (9,0) -- (10,0);



\draw (0,0.4) node[anchor = south] {{\color{blue}\small $\frac{1}{H_t}$}};
\draw (2.5,0.4) node[anchor = south] {{\color{blue}\small $\frac{1}{2H_t}$}};
\draw (5,0.45) node[anchor = south] {{\color{blue}\small $\frac{1}{3H_t}$}};
\draw (10,0.45) node[anchor = south] {{\color{blue}\small $\frac{1}{tH_t}$}};

\draw (-0.8,-1) node[anchor = north] {{\color{red}\tiny $\frac{1}{H_t}$}};
\draw (-0.2,-1) node[anchor = north] {{\color{red}\tiny $\frac{1}{H_t}$}};
\draw (0.3,-1.3) node {\color{red}\tiny$\cdots$};
\draw (0.7,-1) node[anchor = north] {{\color{red}\tiny $\frac{1}{H_t}$}};
\draw (1.7,-1) node[anchor = north] {{\color{red}\tiny $\frac{1}{2H_t}$}};
\draw (2.3,-1) node[anchor = north] {{\color{red}\tiny $\frac{1}{2H_t}$}};
\draw (2.8,-1.3) node {\color{red}\tiny$\cdots$};
\draw (3.2,-1) node[anchor = north] {{\color{red}\tiny $\frac{1}{2H_t}$}};
\draw (4.2,-1) node[anchor = north] {{\color{red}\tiny $\frac{1}{3H_n}$}};
\draw (4.8,-1) node[anchor = north] {{\color{red}\tiny $\frac{1}{3H_n}$}};
\draw (5.3,-1.3) node {\color{red}\tiny$\cdots$};
\draw (5.7,-1) node[anchor = north] {{\color{red}\tiny $\frac{1}{3H_n}$}};
\draw (9.2,-1) node[anchor = north] {{\color{red}\tiny $\frac{1}{tH_t}$}};
\draw (9.8,-1) node[anchor = north] {{\color{red}\tiny $\frac{1}{tH_t}$}};
\draw (10.3,-1.3) node {\color{red}\tiny$\cdots$};
\draw (10.7,-1) node[anchor = north] {{\color{red}\tiny $\frac{1}{tH_t}$}};
\end{tikzpicture}
	\caption{\em \small An example showing the limitations of good partitions, with a solution to \eqref{lp:fkso:m}-\eqref{lp:fkso:bounds} shown in {\color{red}\textbf{red }($z$ values)} and {\color{blue}blue ($\cov$ values)}. The thin ``edges'' represent distance $1$, the thick ``edges'' represent distance $2$, and all other distances are determined by making triangle inequalities tight. The fault-tolerances are $\ell_{v_1} = 1, \ell_{v_2} = 2, \dots, \ell_{v_t} = t$. \label{fig:part-gap}}
\end{figure}
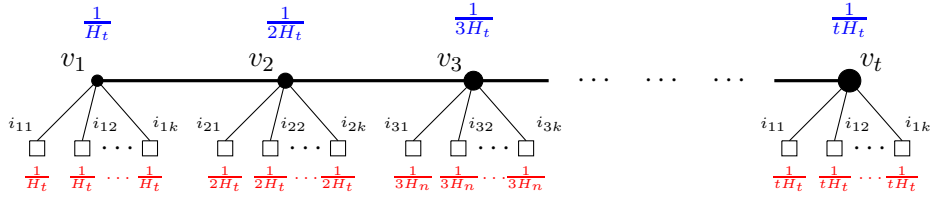
\end{document}